\documentclass{article}

\usepackage[T1]{fontenc}
\usepackage[utf8]{inputenc}
\usepackage{amsmath,amssymb, amsthm}
\usepackage{thmtools}
\usepackage{thm-restate}
\usepackage{hyperref}
\usepackage{cleveref}
\usepackage{color}

\usepackage{lmodern}
\newcommand{\eps}{\varepsilon}
\newcommand{\ver}[1]{}
\newcommand{\oldver}[1]{}
\newcommand{\koz}[1]{}
\declaretheorem[name=Theorem,numberwithin=section]{thm}
\declaretheorem[name=Lemma,numberwithin=section]{lemma}
\declaretheorem[name=Proposition,numberwithin=section]{proposition}
\declaretheorem[name=Definition, numberwithin=section]{definition}

\title{On Slepian--Wolf Theorem with Interaction}

\author{Alexander Kozachinskiy\footnotemark[1]\\[2mm]
\footnotemark[1]~{Moscow State University, Faculty of Mechanics and Mathematics}\\{kozlach@mail.ru}}
\begin{document}
\maketitle
\begin{abstract}
In this paper we study interactive ``one-shot'' analogues of the classical Slepian-Wolf theorem. Alice receives a value of a random variable $X$, Bob receives a value of another random variable $Y$ that is jointly distributed with $X$. Alice's goal is to transmit $X$ to Bob (with some error probability $\varepsilon$). Instead of one-way transmission, which is studied in the classical coding theory, we allow them to interact. They may also use shared randomness.

We show, that Alice can transmit $X$ to Bob in expected $H(X|Y) + 2\sqrt{H(X|Y)} + O(\log_2\left(\frac{1}{\varepsilon}\right))$ number of bits. Moreover, 
we  show that  every one-round protocol $\pi$ with information complexity $I$ 
can be compressed to the (many-round) protocol 
with expected communication about $I + 2\sqrt{I}$ bits.
This improves a result by Braverman and Rao \cite{braverman2011information}, where they had $5\sqrt{I}$. 
Further, we show how to solve this problem (transmitting $X$) using $3H(X|Y) + O(\log_2\left(\frac{1}{\varepsilon}\right))$ bits and $4$ rounds on average. This improves a result of~\cite{brody2013towards}, where they had 
$4H(X|Y) + O(\log1/\varepsilon)$  bits and 10 rounds on average.

In the end of the paper we discuss how many bits Alice and Bob may need to communicate on average besides $H(X|Y)$. The main question is whether the upper bounds mentioned above are tight. We provide 
an example of $(X, Y)$, such that transmission of $X$ from Alice to Bob with error probability $\varepsilon$ requires $H(X|Y) + \Omega\left(\log_2\left(\frac{1}{\varepsilon}\right)\right)$ bits on average.
\end{abstract}

\section{Introduction}
Assume that Alice receives a value of a random variable $X$ and she wants to transmit that value to Bob. It is well-known (\cite{shannon2001mathematical}) that Alice can do it using one message over the binary alphabet of expected length less than $H(X) + 1$. Assume now that there are $n$ independent random variables $X_1, \ldots, X_n$ distributed as $X$, and Alice wants to transmit all 
$X_1, \ldots, X_n$ to Bob. Another classical result from \cite{shannon2001mathematical} states, that Alice can do it using one message of \textit{fixed} length, namely $\approx nH(X)$,
% when $n$ tends to infinity, 
with a small probability of error.

One of the possible ways to generalize this problem is to provide Bob with a value of another random variable $Y$ which is jointly distributed with $X$. That is, to let Bob know some partial information about $X$ for free. This problem is the subject of the classical Slepian-Wolf Theorem~\cite{slepian1973noiseless} which asserts that if there are $n$ independent pairs $(X_1, Y_1), \ldots, (X_n, Y_n)$, each pair distributed exactly as $(X, Y)$, then Alice can transmit all 
$X_1, \ldots, X_n$ to Bob, who knows $Y_1, \ldots, Y_n$, using one message of fixed length, namely $\approx nH(X|Y)$, with a small probability of error.
%when $n$ tends to infinity. 
However, it turns out that a one-shot analogue of this theorem is impossible, if only one-way communication is allowed.

The situation is quite different, if we allow Alice and Bob to \textit{interact}, that is, to send messages in both directions. In \cite{orlitsky1992average} Orlitsky studied this problem for the average-case communication when no error is allowed. He showed that if pair $(X, Y)$ is uniformly distributed on it's support, then Alice may transmit $X$ to Bob using at most
$$H(X|Y) + 3\log_2(H(X|Y) + 1) + 17$$
bits on average and 4 rounds. For the pairs $(X, Y)$ whose 
support is a Cartesian product Orlitsky showed that error-less transmission of $X$ from Alice to Bob requires $H(X)$ bits on average.

From a result of Braverman and Rao (\cite{braverman2011information}), 
it follows that for arbitrary $(X, Y)$ it is sufficient to communicate at most
$$H(X|Y) + 5\sqrt{H(X|Y)} + O\left(\log_2\left(\frac{1}{\varepsilon}\right)\right)$$
bits on average (here $\varepsilon$ stands for the error probability). 
%In their setting Alice's goal was to sample an element from some universe according to the given distribution, and Bob's goal was to output the same element with high probability. Transmission from Alice to Bob is the special case in this setting, when Alice's distribution is trivial (has support of size 1).

We improve this result, showing that Alice may transmit $X$ to Bob with error probability at most $\varepsilon$
(for each pair of inputs) using at most 
$$H(X|Y) + 2\sqrt{H(X|Y)} + O\left(\log_2\left(\frac{1}{\varepsilon}\right)\right)$$
bits on average and $O(\sqrt{H(X|Y)})$ rounds. Our protocol is inspired by protocol from \cite{bauerinternal}. The idea of the protocol is essentially the same, we only apply some technical trick to reduce communication.

%\color{blue}
Actually, in \cite{braverman2011information} a more general result was established. It was shown there that every one-round protocol $\pi$ with information complexity $I$ can be compressed to the (many-round) protocol with expected length at most 
\begin{equation}
\label{prev_res1}
\approx I + 5\sqrt{I}.
\end{equation}
Using the result from \cite{braverman2014public}, we 
improve \ref{prev_res1}. Namely, we  show that  every one-round protocol $\pi$ with information complexity $I$ can be compressed to the (many-round) protocol with expected communication length at most 
$$\approx I + 2\sqrt{I}.$$

%$$\approx I + \log_2 I + 2\sqrt{I + \log_2 I}.$$

\color{black}

In \cite{brody2013towards}, it is established a one-shot interactive analogue of the Slepian-Wolf theorem for the bounded-round communication. They showed that Alice may transmit $X$ to Bob using at most $O(H(X|Y) + 1)$ bits and $O(1)$ 
rounds on average. More specifically, their protocol
transmits at most $4H(X|Y) +\log_2(1/\eps)+O(1)$ 
bits on average in 10 rounds on average. 
In this paper, we provide another proof of this result, 
which seems to be easier. 
More specifically, we show that it is sufficient to communicate at most 
$$3H(X|Y) + \log_2\left(\frac{1}{\varepsilon}\right)+O(1)$$
bits on average in at most $4$ rounds on average.

From the proof of our upper bound it follows that there exists a \textit{deterministic protocol} which transmits $X$ from Alice to Bob using the same number of bits on average (namely $H(X|Y) + 2\sqrt{H(X|Y)} + O\left(\log_2\left(\frac{1}{\varepsilon}\right)\right)$) and which guaranties that for at most $\varepsilon$-fraction of inputs (with respect to the distribution of $(X, Y)$) the transmission is incorrect. Are there random variables 
$X, Y$ for which the corresponding upper bound is tight? 
We make a step towards answering 
this question: we provide an example of random variables $X, Y$ 
such that every deterministic protocol which transmits $X$ from Alice to Bob with error probability $\eps$ 
must communicate at least $H(X|Y) + \Omega\left(\log_2\left(\frac{1}{\varepsilon}\right)\right)$ bits on average.

In the Appendix we provide an example of $(X, Y)$ 
for which it seems plausible that the upper bound
$H(X|Y) + O(\sqrt{H(X|Y)})$ is tight. 
%We explain our intuition and prove a weaker lower bound
%than $H(X|Y) + \Omega(\sqrt{H(X|Y)})$.

\section{Definitions}
We will denote the set of the first $n$ naturals $\{1, 2, \ldots, n\}$ by $[n]$.
\subsection{Information Theory}
Let $X$, $Y$ be two joint distributed random variables, taking values in the finite sets, respectively, $\mathcal{X}$ and $\mathcal{Y}$. 
\begin{definition}
Shannon Entropy of $X$ is defined by the formula
$$H(X) = \sum\limits_{x\in\mathcal{X}}\Pr[X = x]\log_2\left(\frac{1}{\Pr[X = x]}\right).$$
\end{definition}
\begin{definition}
Conditional Shannon entropy of $X$ with respect to $Y$ is defined by the formula:
$$H(X|Y) = \sum\limits_{y\in\mathcal{Y}}H(X|Y = y)\Pr[Y = y],$$
where $X|Y = y$ denotes a distribution of $X$, conditioned on the event $\{Y = y\}$.
\end{definition}

If $X$ is uniformly distributed in $\mathcal{X}$ then obviously 
$H(X) = \log_2(|\mathcal{X}|)$. We will also use the fact that 
the formula for conditional entropy may be re-written as
$$H(X|Y) = \sum\limits_{(x, y)\in\mathcal{X}\times\mathcal{Y}}\Pr[X = x, Y = y]\log_2\left(\frac{1}{\Pr[X = x|Y = y]}\right).$$
Generalization of the Shannon entropy is Renyi entropy.
\begin{definition}
Renyi entropy of $X$ 
%of order $\alpha \ge 0$, $\alpha\neq 1$ 
is defined by the formula
$$H_2(X) =  -\log_2\left(\sum\limits_{x\in\mathcal{X}}\Pr[X = x]^2\right).$$
%$$H_\alpha(X) = \frac{1}{1 - \alpha}\log_2\left(\sum\limits_{x\in\mathcal{X}}\Pr[X = x]^\alpha\right).$$
\end{definition}
Concavity of $\log$ implies that 
%if $\alpha > 1$, then 
$H(X) \ge H_2(X)$.
 
%\color{blue}
The mutual information of two random variables $X$ and $Y$, conditioned on another random variable $Z$, can be defined as:
$$I(X:Y|Z) = H(X|Z) - H(X|Y, Z).$$
\color{black}

For the further introduction in information theory see, for example \cite{yeung2008information}.

\subsection{Communication Protocols}
%\color{blue}Primarily we consider communication protocols to solve transmission problems.
\color{black}
Assume that we are given jointly distributed random variables $X$ and $Y$, taking values in finite sets $\mathcal{X}$ and $\mathcal{Y}$.
Let $R, R_A, R_B$ be a random variables, taking values in finite sets $\mathcal{R}$, $\mathcal{R_A}$ and $\mathcal{R_B}$, such that  $(X, Y), R, R_A, R_B$ are mutually independent.

\begin{definition} A randomized communication protocol 
is a rooted binary tree, in which each non-leaf vertex is associated either with Alice or with Bob. For each non-leaf vertex $v$ associated with Alice there is a function 
$f_v:\mathcal{X}\times\mathcal{R}\times\mathcal{R_A}\to\{0, 1\}$ and for each non-leaf vertex $u$ associated with Bob there is a function $g_u:\mathcal{Y}\times\mathcal{R}\times\mathcal{R_B}\to\{0, 1\}$. For each non-leaf vertex one of an out-going edges is labeled by $0$ and other is labeled by $1$. Finally, for each leaf $l$ there is a function $\phi_l:\mathcal{Y}\times\mathcal{R}\times\mathcal{R_B}\to\mathcal{O}$, where $\mathcal{O}$ denotes the set of all possible Bob's outputs.
\end{definition}

 \koz{We could define the Alice's output in the same manner, but in all the protocols below we are interested only in Bob's output.}

A computation according to a protocol runs as follows. Alice is given $x \in \mathcal{X}$, Bob is given $y \in \mathcal{Y}$. Assume that the random variables $R$ takes a value $r$, $R_A$ takes a value $r_a$ and $R_B$ takes a value $r_b$. Alice and Bob start at the root of the tree. If they are in the non-leaf vertex $v$ associated with Alice, then Alice sends $f_v(x, r, r_a)$ to Bob and they go by the edge labeled by $f_v(x, r, r_a)$. If they are in a non-leaf vertex associated with Bob then Bob sends $g_v(y, r, r_b)$ to Alice and they go by the edge labeled by $g_v(y, r, r_b)$. When they reach a leaf $l$ Bob outputs the result $\phi_l(y, r, r_b)$.

A protocol is called \textit{public-coin} if $f_v, g_u$ and $\phi_l$ do not depend on the values of $R_A, R_B$.

A protocol is called \textit{deterministic} if $f_v, g_u$ and $\phi_l$ do not depend on the values of $R, R_A, R_B$.

We distinguish between average-case communication complexity and the worst-case communication complexity.

\begin{definition}
The (worst-case) communication complexity of a protocol $\pi$, denoted by $CC(\pi)$, is defined as the depth of the corresponding binary tree.

We say that protocol $\pi$ communicates $d$ bits on average (or expected length of the protocol is equal to $d$), if the expected depth of the leaf that Alice and Bob reach during the execution of the protocol $\pi$ is equal to $d$, where the expectation is taken over $X$, $Y$, $R$, $R_A$, $R_B$.
\end{definition}

If the Alice's goal is to transmit $X$ to Bob, then in the end of the communication Bob should output some element of $\mathcal{X}$ (that is, $\mathcal{O} = \mathcal{X}$). We say that protocol transmits $X$ from Alice to Bob with error probability $\eps$ if
$$\Pr[X = \phi_L(Y, R, R_B)] \ge 1 - \eps,$$
where $L$ denotes the leaf that Alice and Bob reach in the protocol tree.

For the worst-case communication it is sufficient to consider only deterministic protocols. Indeed, assume that we are given a randomized protocol solving our problem with error 
probability $\eps$. Fix the value of $R$ for which error probability is minimal. In this way we obtain a protocol with the same worst-case communication complexity and error probability.

For the further introduction in Communication Complexity see \cite{kushilevitz2006communication}

\section{Near-optimal one-shot Slepian-Wolf theorem}

Consider the following auxiliary problem. Let $A$ be a finite set. 
Assume that Alice receives an arbitrary $a\in A$ and Bob receives and arbitrary probability distribution $\mu$ on $A$.
Alice wants to communicate $a$ to Bob in about $\log (1/\mu(a))$ bits
with small probability of error.

\begin{lemma}\label{u1} Let $\eps$ be a positive real and 
$h$ a positive integer. There exists a
public coin randomized communication protocol such that
for all $a$ in the support of $\mu$ the following hold:
\begin{itemize}

\item in the end of the communication 
Bob outputs $b\in A$ which is equal to $a$ with probability at least $1 - \eps$;

\item the protocol communicates at most 
$$\log_2\left(\frac{1}{\mu(a)}\right) + \frac{\log_2\left(\frac{1}{\mu(a)}\right)}{h} + h +\log_2\left(\frac{1}{\eps}\right) + O(1)$$
bits, regardless of the randomness.

\end{itemize}
\end{lemma}
\begin{proof}
Alice and Bob interpret each portion of $|A|$ consecutive 
bits from the public randomness source as a table of 
a random function $h:A\rightarrow \{0,1\}$. That is, we will think 
that they have
access to a large enough family of mutually independent random functions of the type $A\rightarrow \{0,1\}$. Those functions will be called \emph{hash functions} and their values \emph{hash values} 
below.

The first set $k = \left\lceil\log_2\left(\frac{1}{\eps}\right)\right\rceil + 1$. Then Bob sets:
$$S_i = \left\{x\in A\,|\,\mu(x) \in (2^{-i - 1}, 2^{-i}]\right\}.$$

Then Alice and Bob work in stages numbered $0,1,\dots$.

On \emph{Stage 0:} 
\begin{enumerate}
\item Alice sends $k$ hash values of $a$ to Bob.
\item Bob computes set $S_0^\prime$, which consists of all elements 
from $S_0$ that have the same hash values as sent by Alice
(actually $S_0$ has at most one element).
\item If $S_0^\prime \neq \varnothing$, then Bob sends $1$ to Alice, outputs any element of $S_0^\prime$ and they terminate. Otherwise Bob sends $0$ to Alice and they proceed of Stage 1.
\end{enumerate}

On \emph{Stage $t$}:
\begin{enumerate}
\item Alice sends $h$ new hash values of $a$ to Bob so that 
the total number of hash values of $a$ available to Bob be $k + ht$. 
\item For each $i \in \{h(t - 1) + 1, \ldots, ht\}$ Bob computes set $S_i^\prime$, which consists of all elements from $S_i$, which agree with all Alice's hash values.
\item If there exists $i \in \{h(t - 1) + 1, \ldots, ht\}$ such that $S_i^\prime \neq \varnothing$, then Bob sends $1$ to Alice, outputs any element of $S_i^\prime$ and they terminate. 
Otherwise Bob sends $0$ to Alice and they proceed to Stage $t + 1$.
\end{enumerate}

Let us at first show that the protocol terminates for all $a$ in the support
of $\mu$. Assume that Alice has $a$ and Bob has $\mu$. 
Let $i = \left\lfloor\log_2\left(\frac{1}{\mu(a)}\right)\right\rfloor$ so that $a\in S_i$. 
The protocol terminates on Stage $t$ where 
$$h(t - 1) + 1 \le i \le ht$$
or earlier. Indeed all hash values of $a$ available to Bob on
Stage $t$ coincide 
with hash values of some element of $S_i$ (for instance, with those of $a$).

Thus Alice sends at most $k + ht$ bits to Bob and Bob sends at most $1 + t$ bits to Alice. 
Therefore total communication is bounded by 
$$
\begin{aligned}
k + ht + 1 + t &= k + h(t - 1) + h + 2 + (t - 1)\\
&\le k + i - 1 + h + 2 + \frac{i - 1}{h}\\
&\le k + \log_2\left(\frac{1}{\mu(a)}\right) + \frac{\log_2\left(\frac{1}{\mu(a)}\right)}{h} + h + O(1).
\end{aligned}
$$
Since $k = \left\lceil\log_2\left(\frac{1}{\eps}\right)\right\rceil + 1$, the required bound follows.

Now we bound the error probability. An error may occurs, if for some $t$ a set 
$S_i$ considered on Stage $t$ has an element $b\neq a$ which agrees with hash values sent from Alice.
At that time Bob has already $k+ht\ge k+i$ hash values. The probability that 
$k+i$ hash values of $b$ coincide with those of $a$ is $2^{-k-i}$. Hence
by union bound error probability does not exceed
$$
\begin{aligned}
\sum\limits_{i = 0}^{\infty} |S_i|2^{-k - i} &= 
2^{-k + 1}\sum\limits_{i = 0}^\infty |S_i| 2^{-i - 1}
< 2^{-k + 1}\sum\limits_{i = 0}^\infty\sum\limits_{x\in S_i}\mu(x)\\
&= 2^{-k + 1}\sum\limits_{x\in A}\mu(x)= 2^{-k + 1} = 2^{- \left\lceil\log_2\left(\frac{1}{\eps}\right)\right\rceil} \le \eps.
\end{aligned}
$$
\end{proof}

\begin{thm}\label{thm:main}
Let $X$, $Y$ be jointly distributed random variables that take values in the finite sets 
$\mathcal{X}$ and $\mathcal{Y}$. 
%Assume that Alice receives a value of $X$ 
%and Bob receives a value of $Y$. 
Then for every positive $\eps$ there exists a 
public-coin protocol with the following properties.
\begin{itemize}
\item For every pair $(x, y)$ from the support of $(X,Y)$ with probability at least 
$1 - \eps$ Bob outputs $x$;
\item The expected length of communication is at most
$$H(X|Y) + 2\sqrt{H(X|Y)} + \log_2\left(\frac{1}{\eps}\right) + O(1).$$
\end{itemize}
\end{thm}

\begin{proof}
On input $x,y$, Alice and Bob run protocol of Lemma \ref{u1} 
with $A=\mathcal{X}$, $h = \left\lceil \sqrt{H(X|Y)}\right\rceil$, 
$a = x$ and $\mu$ equal to the distribution of $X$, conditioned on the event $Y = y$. 
Notice that Alice knows $a$ and Bob knows $\mu$.

Let us show that both requirements are fulfilled for this protocol.
The first requirement immediately follows from the first property
of the protocol of Lemma \ref{u1}.

From the second property
of the protocol of Lemma \ref{u1} it follows that for 
input pair $x,y$ out protocol communicates at most:
$$
\log_2\left(\frac{1}{\Pr[X = x|Y = y]}\right) + \frac{\log_2\left(\frac{1}{\Pr[X = x|Y = y]}\right)}{\left\lceil \sqrt{H(X|Y)}\right\rceil} + \left\lceil \sqrt{H(X|Y)}\right\rceil +\log_2\left(\frac{1}{\eps}\right) + O(1)
$$
bits. Recalling that 
$$
H(X|Y)=\sum\limits_{(x, y)\in\mathcal{X}\times\mathcal{Y}}\Pr[X = x, Y = y]\log_2\left(\frac{1}{\Pr[X = x|Y = y]}\right) 
$$
we see on average the communication is as short as required. 
\end{proof}

\textit{Remark.} 
One may wonder whether there exists a private-coin communication protocol
with the same properties as the protocol of Theorem \ref{thm:main}. 
Newman's theorem (\cite{newman1991private}) states that every 
public-coin protocol can be transformed into 
a private-coin protocol at the expense of increasing 
the error probability by $\delta$ and the worst case communication 
by $O(\log\log|\mathcal{X}\times\mathcal{Y}|+\log1/\delta)$ (for any positive $\delta$).
Lemma \ref{u1} provides an upper bound for 
the error probability and communication of our protocol for each pair of inputs. 
Repeating the arguments from the proof of Newman's theorem, 
we are able to transform the public-coin protocol of Lemma~\ref{u1}
into a private-coin one with the same trade off between the increase of error probability
and the increase of communication length.
It follows that for our problem there exists a 
private-coin communication protocol which errs with probability at most $\eps$ and 
communicates on average as many bits as the public-coin protocol from Theorem \ref{thm:main} 
plus extra $O(\log\log|\mathcal{X}\times\mathcal{Y}|)$ bits.

\section{One-shot Slepian-Wolf theorem with a constant number of rounds on average}
In this section, 
we modify the construction from the previous section 
to reduce the average number of rounds to a constant.

\begin{thm}
Let $X$, $Y$ be jointly distributed random variables that take values in the finite sets 
$\mathcal{X}$ and $\mathcal{Y}$. 
Then for every positive $\eps$ there exists a 
public-coin protocol with the following properties:
\begin{itemize}
\item For every pair $(x, y)$ from the support of $(X,Y)$ with probability at least 
$1 - \eps$ Bob outputs $x$;
\item The expected length of the protocol does not exceed
$$3H(X|Y) + \log_2\left(\frac{1}{\eps}\right) + O(1).$$
\item The expected number of rounds in protocol is at most 4.
\end{itemize}
(Compared to Theorem~\ref{thm:main}, the number of rounds has decreased and the communication
length has increased.)
\end{thm}

\begin{proof}
We will use the following notation:
$$
l = \lceil H(X|Y)\rceil, \qquad
k = \left\lceil\log_2\left(\frac{1}{\eps}\right)\right\rceil+1,
$$ 
$$\mu(x, y) = \Pr[X = x, Y = y],\qquad
\mu(x|y) = \Pr[X = x|Y = y].$$
%Recall that 
%$$
%H(X|Y) = \sum\limits_{(x, y)\in \mathcal{X}\times \mathcal{Y}}\mu(x, y)\log_2\left(\frac{1}{\mu(x|y)}\right).$$

Alice and Bob apply the following modification 
of the protocol of Lemma~\ref{u1}.
Recall that that protocol works in stages. On Stage 0 Alice sends to Bob 
$k$ random hash bits and on each subsequent stage 
Alice sends to Bob extra $h$ random hash bits. On each stage Bob 
looks for an 
element in all sets 
$$
S_i = \{x^\prime\left.\right| \mu(x^\prime|y) \in (2^{-i - 1}, 2^{-i}]\}.
$$
such that $i$ is at least $k$ less than the total number of hash 
bits he has so far. 
This guarantees that the error probability is at least $\eps$ for all input pairs.

Now Alice sends $k+l$ hash bits on Stage 0 and 
$l2^t$ new hash bits on Stage $t>0$. This is the main difference between the new protocol and
the protocol of Theorem~\ref{thm:main}. In order to keep the error
probability at most $\eps$,
on Stage $t$ Bob looks for an element
in $S_i$ with the same hash values as sent by Alice
for $i\le l + 2l + \ldots + 2^tl$. 
If there is such an element, then Bob 
outputs any such element (and sends 1 to Alice). Otherwise he sends 0 and they proceed
to the next stage.

As earlier, 
by union bound the error probability does not exceed
$$
\begin{aligned}
\sum\limits_{i = 0}^{\infty} |S_i|2^{-k - i} &= 
2^{-k+1}\sum\limits_{i = 0}^\infty |S_i| 2^{-i - 1}
\le 2^{-k+1}\sum\limits_{i = 0}^\infty\sum\limits_{x^\prime\in S_i}\mu(x^\prime|y)\\
&= 2^{-k+1}\sum\limits_{x^\prime\in\mathcal{X}}\mu(x^\prime|y)\le 2^{-k+1} = 2^{- \left\lceil\log_2\left(\frac{1}{\eps}\right)\right\rceil} \le \eps.
\end{aligned}
$$

%Denote the maximum number of bits sent from Alice when $X = x$, $Y = y$ by $A_{xy}$ 
%and denote the maximum number of bits sent from Bob when $X = x$, $Y = y$ by $B_{xy}$. It is easy to see that:
%\begin{equation}
%\label{AB}
%lB_{xy}\le A_{xy} - k.
%\end{equation}

Now we will estimate the communication length on each input pair $(x,y)$
of positive probability. Bob sends one bit in each round. As we will see
the average number of rounds is at most 4, thus we may forget about the communication 
from Bob and concentrate on communication from Alice.

Set $j = j(x, y) = \left\lfloor\log_2\left(\frac{1}{\mu(x|y)}\right)\right\rfloor.$
Notice that $x\in S_j$. 
%Assume first that $j \ge l$ and 
Consider 
$t$ such that
\begin{equation}
\label{main}
l + 2l + \ldots + 2^{t - 1}l <j \le l + 2l + \ldots + 2^{t}l .
\end{equation}
By the construction of the protocol
the communication length for input $x,y$ is at most
$$
\begin{aligned}
&k + l + 2l + \ldots + 2^{t}l\\
&=
k + l + 2(l + 2l + \ldots + 2^{t - 1}l)\\
&<k + l + 2j.
\end{aligned}
$$
%We proved, that if $j \ge l$, then $A(x, y) \le k + l + 2j$. But if $j < l$, 
%we have that $A(x, y) \le k + l$. Since $j\ge 0$, for each pair $(x, y)$ we at least have:
%$$A(x, y) \le k + l + 2j.$$
%We may bound an average number of bits sent from Alice in the the following way:
Hence the expected length of communication
from Alice to Bob is at most
$$
\begin{aligned}
%\sum\limits_{(x, y)\in \mathcal{X}\times\mathcal{Y}}\mu(x, y)A(x, y)&\le 
&\sum\limits_{(x, y)\in \mathcal{X}\times\mathcal{Y}}\mu(x, y)(k + l + 2j(x, y))\\
&\le k + l + 2\sum\limits_{(x, y)\in \mathcal{X}\times\mathcal{Y}}\mu(x, y)\log_2\left(\frac{1}{\mu(x|y)}\right) \\
&= k + l + 2H(X|Y) = 3H(X|Y) + k + O(1).
\end{aligned}
$$

Let us bound the expected number of rounds in our protocol. 
Let $R(x, y)$ stand for the number for inputs $X = x$, $Y = y$. 
%Suppose for a moment that $j(x, y) \ge l$. 
Then $R(x, y)$ is at most $2t + 2$, where $t$ is defined by \eqref{main}. 
By \eqref{main} we have 
$$(2^{t} - 1)l < j \le \log_2\left(\frac{1}{\mu(x|y)}\right)$$
and hence:
$$t \le \log_2\left(1 + \frac{\log_2\left(\frac{1}{\mu(x|y)}\right)}{l}\right).$$
Thus:
$$
R(x, y) \le 2 + 2\log_2\left(1 + \frac{\log_2\left(\frac{1}{\mu(x|y)}\right)}{l}\right).
$$
By concavity of the logarithmic function 
the average number of rounds does not exceed:
$$
\begin{aligned}
\sum\limits_{(x, y)\in \mathcal{X}\times\mathcal{Y}}\mu(x, y)R(x, y)&\le \sum\limits_{(x, y)\in \mathcal{X}\times\mathcal{Y}}\mu(x, y)\left(2 + 2\log_2\left(1 + \frac{\log_2\left(\frac{1}{\mu(x|y)}\right)}{l}\right)\right)\\
&\le 2 + 2\log_2\left(\sum\limits_{(x, y)\in \mathcal{X}\times\mathcal{Y}}\mu(x, y)\left(1 + \frac{\log_2\left(\frac{1}{\mu(x|y)}\right)}{l}\right)\right)\\
&\le 2 + 2\log_2(2) = 4.
\end{aligned}
$$
\end{proof}

%\color{blue}
 \section{One-round Compression}

Information complexity of the protocol $\pi$ with inputs $(X, Y)$  is defined as 
\begin{align*}
IC_\mu(\pi)  &= I(X : \Pi|Y, R) + I(Y : \Pi|X, R)\\
&= I(X : \Pi|Y, R,R_B) + I(Y : \Pi|X, R,R_A)\\
&= I(X : \Pi,R,R_B|Y) + I(Y : \Pi,R,R_A|X),
\end{align*}
where $R,R_A,R_B$ denote (shared, Alice's and Bob's)
randomness,
$\mu$ stands for the distribution of $(X, Y)$ and $\Pi$ stands for the concatenation of all bits sent in $\pi$ ($\Pi$ is called a \emph{transcript}). The first term is equal to the information which Bob learns about Alice's input and the second term is equal to the information which Alice learns about Bob's input. Information complexity is an important concept in the Communication Complexity. For example, information complexity plays the crucial role in the Direct-Sum problem (\cite{weinstein2015information}). 

We will
consider the special case when $\pi$ is \emph{one-round}. In this case Alice sends one message $\Pi$ to Bob, then Bob outputs
the result (based on his input, his randomness, and Alice's message) and
the protocol terminates. 
Since Alice learns nothing, information complexity can be re-written as
$$I = IC_\mu(\pi) = I(X:\Pi|Y, R).$$

Our 
goal is to simulate a given one-round
protocol $\pi$ 
with another protocol $\tau$ which has the same input space $(X, Y)$ and whose expected communication complexity is close to $I$. 
The new protocol $\tau$ may be many-round. The quality of simulation will be
measured by the  statistical distance. Statistical distance
 between random variables $A$ and $B$, 
both taking values in the set $V$, equals
$$\delta(A, B) = \max\limits_{U\subset V} \left|\Pr[A \in U] - \Pr[B \in U]\right|.$$
One of the main results of \cite{braverman2011information} 
is the following theorem.
\begin{thm}
\label{one_round_compression}
For every one-round protocol $\pi$ and for every 
probability distribution $\mu$
there is a public-coin protocol $\tau$ with expected length 
(with respect to $\mu$ and the randomness of $\tau$) at most 
$I +  5\sqrt{I} + O\left(\log_2\frac{1}{\varepsilon}\right)$ such that 
for each pair of inputs $(x, y)$ after termination
of $\tau$ Bob outputs a random variable 
$\Pi^\prime$ with $\delta\left(\left(\Pi|X = x, Y = y\right), \left(\Pi^\prime|X = x, Y = y\right)\right) \le \varepsilon$.
\end{thm}

We will show that theorem \ref{thm:main} implies that
we can replace $5\sqrt{I}$ by about $2\sqrt{I}$ in this theorem.
We want transmit Alice's message $\Pi$ to Bob 
(who knows $Y$ and his randomness $R$) in many rounds
so that the expected communication length is small. 
By theorem \ref{thm:main} this
task can be solved with error $\varepsilon$ in expected communication 
\begin{equation}\label{eq-p}
H(\Pi|Y,R) + 2\sqrt{H(\Pi|Y,R)} + O\left(\log_2 \frac{1}{\varepsilon}\right).
\end{equation}

Assume first that the original protocol
$\pi$ uses only public randomness.
Then 
$$I = I(X:\Pi|Y, R) = H(\Pi|Y, R) - H(\Pi|X, Y, R) = H(\Pi|Y, R).$$
Indeed, $H(\Pi|X, Y, R) = 0$, since $\Pi$ is defined by $X, R$.
Thus~\eqref{eq-p} becomes 
 $$
I + 2\sqrt{I} + O\left(\log_2 \frac{1}{\varepsilon}\right)
$$
and we are done.

\oldver{In general case, when the original protocol
uses private randomness, $I$ can be much smaller 
than $H(\Pi|Y, R)$.}  
Fortunately, by the following theorem from \cite{braverman2014public} we can remove private coins from the protocol with only a slight increase in information complexity.
\begin{thm}
\oldver{For every one-round protocol $\pi$ and for every 
probability distribution $\mu$}
there is a one-round public-coin protocol 
$\pi^\prime$ with information complexity 
$IC_\mu(\pi) \le I + \log_2 (I + O(1))$ 
such that for each pairs of inputs $(x, y)$ 
\oldver{in the protocol $\pi'$} Bob outputs $\Pi^\prime$ for which $\Pi^\prime|X =x, Y = y$ and $\Pi|X = x, Y = y$ are identically distributed. 
\end{thm}

Combining this theorem with our main result (theorem \ref{thm:main}), we obtain the following theorem.
\begin{thm}
\label{our_one_round_compression}
\oldver{For every one-round protocol $\pi$ and for every 
probability distribution $\mu$}
there is a public-coin protocol $\tau$ with expected length 
(with respect to $\mu$ and the randomness of $\tau$) at most $$I + \log_2(I + O(1)) +  2\sqrt{I + \log_2(I + O(1))} + O\left(\log_2\frac{1}{\varepsilon}\right)$$ such that for each pair of inputs $(x, y)$ in \oldver{the protocol $\tau$} 
Bob outputs $\Pi^\prime$ \oldver{with} 
$\delta\left(\left(\Pi|X = x, Y = y\right), \left(\Pi^\prime|X = x, Y = y\right)\right) \le \eps$
\end{thm}

%Compared to theorem \ref{one_round_compression}, 
%we have reduced the constant before the square root from 5 to 2.
\color{black}
\section{A Lower Bounds for the Average-Case Communication}

Let $(X, Y)$ be a pair of jointly
distributed random variables. Assume that $\pi$ is a deterministic protocol
to transmit $X$ from Alice to Bob who knows $Y$. 
Let $\pi(X, Y)$ stand for the result output
by the protocol $\pi$ for input pair $(X,Y)$. 
We assume that for at least $1-\eps$ input pairs
this result is correct:
$$
\Pr[\pi(X, Y)\neq X)]\le\varepsilon.
$$ 
It is not hard to see that in this
case the expected communication length 
cannot be much less than $H(X|Y)$ bits on average.
Moreover, this applies for communication from Alice
to Bob only. 
\begin{proposition}
\label{Fano_bound}
For every deterministic 
protocol as above the expected communication from Alice
to Bob is at least 
$H(X|Y) - \varepsilon\log_2|\mathcal{X}| - 1.$ 
\end{proposition}
\begin{proof}
Indeed, let $\Pi_A$ denote the concatenation of all bits sent by Alice. 
If Bob's input is fixed, then the set of all possible values of 
$\Pi_A$ forms a prefix-free code. Hence
$$\mathsf{E}\left[|\Pi_A| \left.\right| Y = y\right]\ge H(\Pi_A|Y = y)$$
and therefore
$$\mathsf{E}|\Pi_A| = \mathsf{E}_{y\sim Y}\mathsf{E}\left[|\Pi_A| \left.\right| Y = y\right] \ge \mathsf{E}_{y\sim Y} H(\Pi_A|Y = y) = H(\Pi_A|Y).$$
Consider $I(X:\Pi_A|Y)$. By definition $I(X:\Pi_A|Y) \le H(\Pi_A|Y)$. On the other hand we have
$$I(X:\Pi_A|Y) = H(X|Y) - H(X|Y, \Pi_A).$$
Notice that
$\pi(X, Y)$ is a function of $Y$ and $\pi_A$ 
(Bob's guess is based on $Y$ and on bits received from Alice)
and hence $H(X|Y, \Pi_A) \le 
H(X|\pi(X,Y))$. 
Since $\Pr[\pi(X, Y) \neq X]\le \varepsilon$, from Fano inequality it follows that
$$ 
H(X|\pi(X,Y))\le
1 + \varepsilon\log_2|\mathcal{X}|.$$
Therefore $\mathsf{E}|\Pi_A| \ge H(X|Y) - \varepsilon\log_2|\mathcal{X}| - 1.$ 
\end{proof}

There are random variables for which
this lower bound is tight. For instance,
let $Y$ be empty and let 
$X$ take the value $x\in\{0,1\}^n$
with probability $\eps/2^n$ (for all such $x$) 
and let $X=$ (the empty string)
with the remaining 
probability $1-\eps$.
Then the trivial protocol with no communication
solves the job with error probability $\eps$
and $H(X|Y)\approx \varepsilon\log_2|\mathcal{X}|$.

In this section we consider the following question:
are there a random variables $(X, Y)$, for which 
for every deterministic 
communication protocol the expected communication
is significantly larger than $H(X|Y)$, say close to the upper bound 
$H(X|Y) + 2\sqrt{H(X|Y)} + \log_2\left(\frac{1}{\eps}\right)$
of Theorem~\ref{thm:main}? 
Notice that from the proof of the theorem \ref{thm:main} 
it follows that there exists a \textit{deterministic protocol} 
which transmits $X$ from Alice to Bob using 
$H(X|Y) + 2\sqrt{H(X|Y)} + 
O\left(\log_2\left(\frac{1}{\varepsilon}\right)\right)$ bits on average 
and which guaranties that for at most $\varepsilon$-fraction of inputs 
(with respect to the distribution of $(X, Y)$) the transmission is incorrect. 
Indeed, for any choice of randomness the communication on each pair 
of inputs is bounded by lemma \ref{u1}. Thus we may fix random bits
so that the error probability is at most $\eps$.

Orlitsky showed that if no error is allowed
and the support of $(X, Y)$ is a Cartesian product, 
then every deterministic protocol 
must communicate $H(X)$ bits on average. 
\begin{lemma}\label{l:orlitsky}
Let $(X, Y)$ be a pair of jointly 
distributed random variables whose support is a Cartesian product. 
Assume that $\pi$ is a deterministic protocol, which transmits $X$ from Alice 
to Bob who knows $Y$ and
$$\Pr[\pi(X, Y) \neq X)] = 0.$$
Then the expected length of $\pi$ is at least $H(X)$.
\end{lemma}
For the sake of completeness we provide a proof of this result in the Appendix.
%The question is whether this bound is tight. We make the first step in studying this question. 
The main result of this section states 
that there are random variables
$(X, Y)$ such that transmission of $X$ from Alice to Bob with 
error probability $\eps$ requires 
$H(X|Y) + \Omega\left(\log_2\left(\frac{1}{\varepsilon}\right)\right)$ 
bits on average.

The random variables $X,Y$ are specified by two parameters,
$\delta\in(0, 1/2)$ and $n\in\mathbb{N}$. 
Both random variables take values in 
$\{0, 1, \ldots, n\}$ and are distributed
as follows: $Y$ is distributed 
uniformly in $\{0, 1, \ldots, n\}$ and 
$X=Y$ with probability $1-\delta$ and 
$X$ is uniformly distributed in $\{0, 1, \ldots, n\}\setminus\{X\}$
with the remaining probability $\delta$. That is,
$$
\Pr[X = i, Y = j] = \frac{(1 - \delta)\delta_{ij} + \frac{\delta}{n}(1 - \delta_{ij})}{n + 1},
$$
where $\delta_{ij}$ stands for the Kronecker's delta. 
Notice that $X$ is uniformly distributed on $\{0, 1, \ldots, n\}$ as well.
A straightforward calculation reveals that 
$$\Pr[X = i|Y = j] = \frac{\Pr[X = i, Y = j]}{\Pr[Y = j]} = 
(1 - \delta -\frac\delta n)\delta_{ij} + \frac{\delta}{n}$$
and
$$H(X|Y) = (1 - \delta)\log_2\left(\frac{1}{1 - \delta}\right) + \delta\log_2\left(\frac{n}{\delta}\right) = \delta\log_2 n + O(1).$$

We will think of $\delta$ as a constant, say $1/4$. For
one-way protocol we are able to show that 
communication length must be close to $\log n$, which 
is about $1/\delta$ times larger than $H(X|Y)$:
\begin{proposition}
Assume that $\pi$ is a one-way deterministic protocol, 
which transmits $X$ from Alice to Bob who knows $Y$ and
$$\Pr[\pi(X, Y) \neq X)] \le \varepsilon.$$
Then the expected length of $\pi$ is at least 
$\left(1 - \frac{\varepsilon}{\delta}\right)\log_2(n + 1) - 2$.
\end{proposition}
\begin{proof}
Let $S$ be the number of leafs in $\pi$. For each $j\in\{0, 1, \ldots, n\}$ 
$$\#\left\{i\in\{0, 1, \ldots, n\}\left.\right| \pi(i, j) = i \right\}\le S.$$
Hence the error probability $\eps$
is at least $(n + 1 - S)\frac{\delta}{n}$. 
This implies that 
$$S \ge n\left(1 - 
\frac{\varepsilon}{\delta}\right) + 1\ge (n + 1)\left(1 - \frac{\varepsilon}{\delta}\right).$$
Let $\Pi(X)$ denote the leaf Alice and Bob reach in $\pi$
(since the protocol is one-way, the leaf depends only on $X$).
The expected length of $\Pi(X)$ is at least $H(\Pi)$. 
Let $l_1, l_2, \ldots, l_S$ be the list of all leaves in 
the support of the random variable $\Pi(X)$.
As $X$ is distributed uniformly, we have
$$\Pr[\Pi = l_i]\ge \frac{1}{n + 1}$$
for all $i$. The statement follows from 
\begin{lemma}
\label{entropy_bound}
Assume that $p_1, \ldots, p_k, q_1, \ldots, q_k\in(0, 1)$ satisfy
$$\sum\limits_{i = 1}^kp_i = 1,$$ $$\forall i\in\{1, \ldots, k\}\qquad p_i\ge q_i.$$
Then
$$\sum\limits_{i = 1}^kp_i\log_2\frac{1}{p_i}\ge \sum\limits_{i = 1}^kq_i\log_2\frac{1}{q_i} - 2.$$
\end{lemma}
The proof of this technical lemma is deferred to the Appendix.
The lemma implies that
$$
\begin{aligned}
H(\Pi) &= \sum\limits_{i = 1}^S \Pr[\Pi = l_i]\log_2\left(\frac{1}{\Pr[\Pi = l_i]}\right)\\ &\ge \frac{S}{n + 1}\log_2(n + 1) - 2\ge\left(1 - \frac{\varepsilon}{\delta}\right)\log_2(n + 1) - 2.
\end{aligned}
$$
\end{proof}

The next theorem states that for any fixed $\delta$
every two-way deterministic protocol with error
probability $\eps$ must communicate 
about $H(X|Y)+(1-\delta)\log_2(1/\eps)$ bits on average.
\begin{thm}
Assume that $\pi$ is a deterministic protocol 
which transmits $X$ from Alice o Bob who knows $Y$ and
$$\Pr[\pi(X, Y) \neq X)] \le \varepsilon.$$
Then the expected length of $\pi$ is at least 
$$(1 - \delta - \delta/n)\log_2\left(\frac{\delta}{\varepsilon + \delta/n}\right) + (\delta - 2\varepsilon)\log_2(n + 1) - 2\delta.$$
\end{thm}

The lower bound in this theorem is quite complicated 
and comes from its proof. To understand this bound
assume that $\delta$ is a constant, say $\delta=1/4$,
and $\frac{1}{n}\le\varepsilon \le \frac{1}{\log_2n}$. Then 
$H(X|Y)=(1/4)\log_2 n+O(1)$ and 
the lower bound becomes
\begin{align*}
%(3/4 - 1/4n)\log_2\left(\frac1{\varepsilon}\right) + 
%(1/4)\log_2 n - 2\varepsilon\log_2 n - O(1)\\=
\left(1 - \frac 1 4 - \frac{1}{4n}\right)\log_2\left(\frac{\frac{1}{4}}{\eps + \frac{1}{4n}}\right) + (1/4 - 2\eps)\log_2(n + 1) - \frac{1}{2}\\
\end{align*}
Condition $\frac{1}{n}\le \eps$ implies
that the first term is equal to 
$$(3/4)\log_2\left(\frac{1}{\eps}\right) - O(1).$$ 
Condition $\eps \le \frac{1}{\log_2 n}$ implies that the seconds term 
is equal to $$(1/4)\log_2 n - O(1).$$
Therefore under these conditions the lower bound becomes
$$(1/4)\log_2 n + (3/4)\log_2\left(\frac{1}{\eps}\right) - O(1) = H(X|Y) + (3/4)\log_2\left(\frac{1}{\eps}\right) - O(1).$$
\begin{proof}
Let $\Pi=\Pi(X, Y)$ denote the leaf Alice and Bob reach in the
protocol $\pi$ for input pair $(X,Y)$. As we have seen,
the expected length of communication is at least the entropy $H(\Pi(X,Y))$. 
Let $l_1, \ldots, l_S$ denote all the 
leaves in the support of the random variable $\Pi(X, Y)$. 
The set $\{(x, y)\left.\right| \Pi(x, y) = l_i\}$ 
is a combinatorial rectangle $R_i\subset\{0, 1, \ldots, n\}\times\{0, 1, \ldots, n\}$. Imagine $\{0, 1, \ldots, n\}\times\{0, 1, \ldots, n\}$ as a table 
in which Alice owns columns and Bob owns rows. 
Let $h_i$ be the height of $R_i$ and $w_i$ be the width of $R_i$.
Let $d_i$ stand for the number of diagonal elements in $R_i$ 
(pairs of the form $(j, j)$). 
By definition of $(X, Y)$ we have
\begin{equation}
\label{1}
\Pr[\Pi(X,Y) = l_i] = \frac{(1 - \delta)d_i}{n + 1} + \frac{\delta (h_i w_i - d_i)}{n(n + 1)}.\end{equation}
The numbers
$\{\Pr[\Pi(X,Y) = l_i]\}_{i = 1}^S$ define 
a probability distribution over the set $\{1, 2, \ldots, S\}$
and its entropy equals $H([\Pi(X,Y))$. Equation \eqref{1} 
represents this distribution as a weighted sum of the following
distributions: 
$
\left\{\frac{d_i}{n + 1}\right\}_{i = 1}^S$ and $\left\{\frac{h_i w_i}{(n + 1)^2}\right\}_{i = 1}^S$. That is, Equation \eqref{1} 
implies that
$$\{\Pr[\Pi = l_i]\}_{i = 1}^S = (1 - \delta - \delta/n) \left\{\frac{d_i}{n + 1}\right\}_{i = 1}^S + (\delta + \delta/n)\left\{\frac{h_i w_i}{(n + 1)^2}\right\}_{i = 1}^S.$$

Since entropy is concave, we have
\begin{equation}
\begin{aligned}
H(\Pi) &= H\left(\{\Pr[\Pi = l_i]\}_{i = 1}^S \right) \\&\ge (1 - \delta - \delta/n)H\left(\left\{\frac{d_i}{n + 1}\right\}_{i = 1}^S\right) + (\delta + \delta/n)H\left(\left\{\frac{h_i w_i}{(n + 1)^2}\right\}_{i = 1}^S\right)
\end{aligned}
\end{equation}
The lower bound of the theorem
follows from lower bounds of the 
entropies of these distributions.

\emph{A lower bound for 
$H\left(\left\{\frac{d_i}{n + 1}\right\}_{i = 1}^S\right)$.}
In each row of $R_i$ there is at most 1 element $(x, y)$, 
for which $\pi(x, y) = x$. 
The rectangle $R_i$ consists of
$d_i$ diagonal elements and hence there are at least $d_i^2 - d_i$ 
elements $(x, y)$ in $R_i$ for which $\pi(x, y)\neq x$. Summing over all $i$ we get
$$\varepsilon \ge \sum\limits_{i = 1}^S \frac{\delta(d_i^2 - d_i)}{n(n + 1)}$$
and thus
$$\sum\limits_{i = 1}^S \left(\frac{d_i}{n + 1}\right)^2 \le \frac{\varepsilon + \delta/n}{\delta}.$$
Since Renyi entropy is a lower bound for the Shannon entropy, we have
$$
H\left(\left\{\frac{d_i}{n + 1}\right\}_{i = 1}^S\right) \ge \log_2\left(\frac{1}{\sum\limits_{i = 1}^S \left(\frac{d_i}{n + 1}\right)^2}\right) \ge \log_2\left(\frac{\delta}{\varepsilon + \delta/n}\right).
$$
In $R_i$, there are at most $h_i$ good pairs 
(for which $\pi$ works correctly). 
At most $d_i$ of them has probability $\frac{1 - \delta}{n + 1}$. Hence
$$
\Pr[\Pi = l_i, \pi(X, Y) = X] \le \frac{(1 - \delta)d_i}{n + 1} + \frac{\delta(h_i - d_i)}{n(n + 1)}
$$
and
$$
\begin{aligned}
1 - \varepsilon &\le \Pr[\pi(X, Y) = X] = \sum\limits_{i = 1}^S \Pr[\Pi = l_i, \pi(X, Y) = X]\\
&\le \sum\limits_{i = 1}^S\left(\frac{(1 - \delta)d_i}{n + 1} + \frac{\delta(h_i - d_i)}{n(n + 1)}\right) = 1 - \delta - \delta/n + \frac{\delta}{n(n + 1)} \sum\limits_{i = 1}^S h_i.
\end{aligned}
$$
The last inequality implies that 
$$\sum\limits_{i = 1}^S h_i \ge(1 - \varepsilon/\delta) (n + 1)^2.$$

\emph{A lower bound for 
$H\left(\left\{\frac{h_i w_i}{(n + 1)^2}\right\}_{i = 1}^S\right)$.} 
Since $h_i \le n + 1$, we have
$$
\begin{aligned}
\sum\limits_{i = 1}^S \frac{h_i w_i}{(n + 1)^2}\log_2\left(\frac{(n + 1)^2}{h_i w_i}\right) &\ge \sum\limits_{i = 1}^S \frac{h_i w_i}{(n + 1)^2}\log_2\left(\frac{(n + 1)^2}{(n + 1) w_i}\right)\\
&= -\log_2(n + 1) + \sum\limits_{i = 1}^S h_i \frac{w_i}{(n + 1)^2}\log_2\left(\frac{(n + 1)^2}{w_i}\right).
\end{aligned}
$$
Obviously $\frac{w_i}{(n + 1)^2} \ge \frac{1}{(n + 1)^2}$. By lemma \ref{entropy_bound} we get 
$$
\begin{aligned}
\sum\limits_{i = 1}^S h_i \frac{w_i}{(n + 1)^2}\log_2\left(\frac{(n + 1)^2}{w_i}\right) &\ge \left(\sum\limits_{i = 1}^Sh_i\right)\frac{1}{(n + 1)^2}\log_2\left((n + 1)^2\right) - 2\\ &\ge (2 - 2\varepsilon/\delta)\log_2(n + 1) - 2.
\end{aligned}
$$
Thus
$$
H\left(\left\{\frac{h_i w_i}{(n + 1)^2}\right\}_{i = 1}^S\right) \ge (1 - 2\varepsilon/\delta)\log_2(n + 1) - 2.$$
\end{proof}

\bibliographystyle{acm}
\bibliography{ref}

\begin{thebibliography}{10}

\bibitem{bauerinternal}
{\sc Bauer, B., Moran, S., and Yehudayoff, A.}
\newblock Internal compression of protocols to entropy.

\bibitem{braverman2014public}
{\sc Braverman, M., and Garg, A.}
\newblock Public vs private coin in bounded-round information.
\newblock In {\em Automata, Languages, and Programming}. Springer, 2014,
  pp.~502--513.

\bibitem{braverman2011information}
{\sc Braverman, M., and Rao, A.}
\newblock Information equals amortized communication.
\newblock In {\em Foundations of Computer Science (FOCS), 2011 IEEE 52nd Annual
  Symposium on\/} (2011), IEEE, pp.~748--757.

\bibitem{brody2013towards}
{\sc Brody, J., Buhrman, H., Koucky, M., Loff, B., Speelman, F., and
  Vereshchagin, N.}
\newblock Towards a reverse newman's theorem in interactive information
  complexity.
\newblock In {\em Computational Complexity (CCC), 2013 IEEE Conference on\/}
  (2013), IEEE, pp.~24--33.

\bibitem{kushilevitz2006communication}
{\sc Kushilevitz, E., and Nisan, N.}
\newblock {\em Communication Complexity}.
\newblock Cambridge University Press, 2006.

\bibitem{newman1991private}
{\sc Newman, I.}
\newblock Private vs. common random bits in communication complexity.
\newblock {\em Information processing letters 39}, 2 (1991), 67--71.

\bibitem{orlitsky1992average}
{\sc Orlitsky, A.}
\newblock Average-case interactive communication.
\newblock {\em Information Theory, IEEE Transactions on 38}, 5 (1992),
  1534--1547.

\bibitem{shannon2001mathematical}
{\sc Shannon, C.~E.}
\newblock A mathematical theory of communication.
\newblock {\em ACM SIGMOBILE Mobile Computing and Communications Review 5}, 1
  (2001), 3--55.

\bibitem{slepian1973noiseless}
{\sc Slepian, D., and Wolf, J.~K.}
\newblock Noiseless coding of correlated information sources.
\newblock {\em Information Theory, IEEE Transactions on 19}, 4 (1973),
  471--480.

\bibitem{weinstein2015information}
{\sc Weinstein, O.}
\newblock Information complexity and the quest for interactive compression.
\newblock {\em ACM SIGACT News 46}, 2 (2015), 41--64.

\bibitem{yeung2008information}
{\sc Yeung, R.~W.}
\newblock {\em Information theory and network coding}.
\newblock Springer, 2008.

\end{thebibliography}

\appendix
\section{The proof of Lemma~\ref{l:orlitsky}}

%\begin{proof}
Let $\mathcal{X}\times\mathcal{Y}$ stand for the support of $(X, Y)$. Fix $x\in\mathcal{X}$. Consider the set of all possible leafs Alice and Bob may reach in $\pi$ when $X = x$. Let $l_x$ be the leaf of minimal depth from this set. Denote the depth of $l_x$ by $d(l_x)$. Notice that the expected length of the protocol $\pi$ is at least $\mathsf{E}_{x\sim X} d(l_x)$.

Suppose that for some $x_1, x_2\in\mathcal{X}$, $x_1\neq x_2$ we have $l_{x_1} = l_{x_2}$. It means that there exists $y_1, y_2\in\mathcal{Y}$ such that when $X = x_1, Y = y_1$ and when $X = x_2, Y = y_2$ Alice and Bob reach the same leaf $l_{x_1}$. From the rectangle property it follows that when $X = x_1, Y = y_2$ Alice and Bob reach $l_{x_1}$ too. Hence when $X = x_1, Y = y_2$ and when $X = x_2, Y = y_2$, Bob outputs the same answer, which is contradiction.

Thus $l_x$ defines bijection from the set of all possible values of $X$ to some prefix-free set of binary strings. Hence $\mathsf{E}_{x\sim X} d(l_x)\ge H(X)$.
%\end{proof}

\section{The proof of Lemma~\ref{entropy_bound}}
The 
function $f(x) = x\log_2\frac{1}{x}$ increases on $[0, e^{-1}]$
and its maximum value is $e^{-1}\log_2e<1$. 
Indeed,
$$
f^\prime(x) = \frac{1}{\ln2}(-1 - \ln x) = 
\frac{\ln\left(\frac{1}{ex}\right)}{\ln2}\ge 0$$
when $x\in[0, e^{-1}]$. Since $\sum\limits_{i = 1}^kp_i = 1$, we have
$$\#\left\{i\in\{1, \ldots, k\}\left.\right| p_i > e^{-1}\right\} < e.$$
The left hand 
side of this inequality is an integer
hence $\#\left\{i\in\{1, \ldots, k\}\left.\right| p_i > e^{-1}\right\} \le 2$. 
Thus we conclude
$$
\begin{aligned}
\sum\limits_{i = 1}^kp_i\log_2\frac{1}{p_i} &= \sum\limits_{p_i \le e^{-1}}p_i\log_2\frac{1}{p_i} + \sum\limits_{p_i > e^{-1}}p_i\log_2\frac{1}{p_i}\\
&\ge \sum\limits_{p_i \le e^{-1}}q_i\log_2\frac{1}{q_i} + \sum\limits_{p_i > e^{-1}}0 \\
&\ge \sum\limits_{p_i \le e^{-1}}q_i\log_2\frac{1}{q_i} + \sum\limits_{p_i > e^{-1}}\left(q_i\log_2\frac{1}{q_i} - 1\right) \ge \sum\limits_{i = 1}^kq_i\log_2\frac{1}{q_i} - 2.
\end{aligned}
$$

\section{Random variables, for which Theorem \ref{thm:main} may be tight}

We finish this paper with the example of random variables $(X, Y)$, 
for which we believe that the upper bound from Theorem \ref{thm:main} is tight. Let $H_n$ be the $n$-th harmonic number:
$$H_n = \sum\limits_{k = 1}^n \frac{1}{k} = \ln n + O(1).$$

Let $X$ take values in $\{1, 2, \ldots, n\}$ and $Y$ take values in $S_n$, 
the set of all permutations of the set $\{1,\dots,n\}$. 
The distribution of $X, Y$ is defined as follows:
$$\Pr[X = i, Y = \sigma] = \frac{1}{\sigma(i)H_n n!}.$$
This formula implies that 
%the distribution 
%of $Y$ is uniform 
%and $X|Y = \sigma$ is distributed as follows
%$$\Pr[X = i|Y = \sigma] = \frac{1}{\sigma(i)H_n}.$$ 
%Hence 
$H(X|Y = \sigma)$ does not depend on $\sigma\in S_n$ and 
%we can write
equals
$$\sum\limits_{i = 1}^n\frac{\log_2(iH_n)}{iH_n} = 
\frac{\log_2 n}{2} + O(\log\log n).$$
Thus $H(X|Y)= \frac{\log_2 n}{2} + O(\log\log n)$.

We conjecture that every deterministic protocol, which transmits $X$ from Alice to Bob who knows $Y$ with error probability $\eps < 1/\log_2 n$, communicates at least $\frac{\log_2 n}{2} + \Omega(\sqrt{\log_2(n)})$ bits on average.

\end{document}